\documentclass[10pt, journal]{IEEEtran}

\usepackage{cite}
\usepackage{graphicx}
\usepackage{psfrag}
\usepackage{subfigure}
\usepackage{url}
\usepackage{stfloats}
\usepackage{amsmath}
\usepackage{array}
\usepackage{epsfig}
\usepackage{amssymb}
\usepackage{algorithm}
\usepackage{algcompatible}

\usepackage{xcolor}
\usepackage{amsthm}
\newtheorem{theorem}{\textbf{Theorem}}
\newtheorem{lemma}{\textbf{Lemma}}

\usepackage{mathtools} 

\algnewcommand\algorithmicreturn{\textbf{return}}
\algnewcommand\RETURN{\State \algorithmicreturn}%

\begin{document}

\title{Performance Analysis of Uplink  NOMA-Relevant Strategy Under Statistical Delay QoS Constraints}

\author{Mylene Pischella,~\IEEEmembership{Senior Member,~IEEE}, Arsenia Chorti,~\IEEEmembership{Member,~IEEE} and Inbar Fijalkow,~\IEEEmembership{Senior Member,~IEEE}
\tiny\thanks{Mylene Pischella is with CNAM CEDRIC, France and is a visiting researcher at ETIS, UMR 8051 (contact: mylene.pischella@cnam.fr); Arsenia Chorti and Inbar Fijalkow are with ETIS UMR 8051 CY Cergy Paris Université, ENSEA, CNRS, France.}\normalsize
}

\maketitle

\begin{abstract}
A new multiple access (MA) strategy, referred to as non orthogonal multiple access - Relevant (NOMA-R), allows selecting NOMA when this increases \textit{all} individual rates, i.e., it is beneficial for both strong(er) and weak(er) individual users. This letter provides a performance analysis of the NOMA-R strategy in uplink networks with statistical delay constraints. Closed-form expressions of the effective capacity (EC) are provided in two-users networks, showing that the strong user always achieves a higher EC with NOMA-R. Regarding the network’s sum EC, there are distinctive gains with NOMA-R, particularly under stringent delay constraints.

\vspace{0.1cm}
\end{abstract}

\section{Introduction}
Due to the strict delay requirements of many emerging ultra reliable low latency communication (URLLC) applications in beyond fifth generation (B5G) networks, the investigation of the interplay between statistical delay quality of service (QoS) constraints and wireless propagation conditions is highly timely. In this context, employing the link layer metric of the effective capacity (EC) \cite{WuNegi2003,ECSurvey19} -- which indicates the maximum achievable rate under a target delay-outage probability threshold -- emerges as a natural choice.

In parallel, non-orthogonal multiple access (NOMA) \cite{DaiHanzoSurvery18} has been consistently shown to achieve higher sum spectral efficiencies when compared to OMA or other schemes \cite{Chorti}.  Moreover, NOMA may be required in B5G networks for very large users densities. Up to now, EC analyses in NOMA networks have focused primarily on the downlink \cite{LiuNOMAEC17,YuNOMAEC18, XiaoNOMAEC19}.  With respect to the uplink, in \cite{BelloChorti2019} it was shown that in two-user networks NOMA is more efficient than OMA at low signal to noise ratios (SNRs), whereas the opposite conclusion holds at large SNRs, due to the interference experienced by the strong user. Adaptive multiple access (MA) strategies could therefore enhance the performance; to the best of our knowledge, \cite{PischellaNOMA2019WCL} is the first attempt to propose an adaptive MA strategy, called NOMA-Relevant (NOMA-R). In NOMA-R, clusters of users employ NOMA only when it is beneficial for all of them in terms of their individual rates.  

This letter is the first EC performance analysis of adaptive MA strategies. Its contributions are the following:  (i) we evaluate the probability of using NOMA when NOMA-R is employed; (ii) we use the EC as the performance metric\footnote{We note that  \cite{PischellaNOMA2019WCL} focused on proportional fairness.}
, provide new analytic expressions of the  EC with NOMA-R and compare them to  those derived in the two-users case for NOMA and OMA in \cite{BelloChorti2019}; (iii) we prove that NOMA-R is the strategy that maximizes the EC of the strong user, whereas it  always outperforms OMA but not NOMA for the weak user; (iv) with respect to this latter aspect, 
this loss in EC for the weak user  becomes negligible under stringent delay constraints; (v) numerical results also show that this conclusion holds for a larger number of users. This letter consequently proves that NOMA-R is a very efficient strategy for delay constrained applications.

\section{Probability of using  NOMA with NOMA-R}
\subsection{System model}
Let us consider a network with $K$ users employing either OMA, NOMA or NOMA-R in the uplink. The achievable rate of user $k \in \mathcal{S}_K = \{1,...,K\}$ is denoted in the following by $R_k$,  $\tilde{R}_k$ and $\hat{R}_k$ for NOMA, OMA and NOMA-R, respectively. We assume that the independent and identically distributed (i.i.d.) fading channel coefficients in the links to the base station (BS), denoted by $h_k$, follow unit variance Rayleigh distributions. The channel gains $x_k = |h_k|^2$ are assumed ordered in decreasing order, so that $x_k \leq x_{k+1} \forall k \in \{1,...,K-1\}$, and, their distributions can be found by using the theory of order statistics
\cite{yang_alouini_2011}.  We denote by $\rho=\frac{1}{N}$ the transmit SNR with 
$N$ the additive white Gaussian noise power in each link (assumed the same for all links for simplicity). The transmit power used by user $k$ is denoted by $P_k$ 
so that its received SNR is $\rho P_k$, $k \in \mathcal{S}_K$. Let us assume that a cluster $\mathcal{S}$ of users in $ \mathcal{S}_K$ with cardinality $|\mathcal{S}|$ is chosen for NOMA.  The achievable rates (in bits/s/Hz) of the $k$th user in $\mathcal{S}$, assuming perfect successive interference cancellation (SIC) decoding, can be expressed as:
\begin{align} \label{RatesNOMA}
   R_k   =  \frac{|\mathcal{S}|}{K}\log_2\left(1+ \frac{\rho P_k x_{k}}{1+\sum_{\substack{j \in \mathcal{S}\\j<k}}\rho P_j x_{j}} \right), \forall k \in \mathcal{S}.
\end{align}
On the other hand, for the $k$th user in $\mathcal{S}_K\setminus \mathcal{S} $ the achievable rates with OMA are given as  \begin{align} \label{RatesOOMA}
  \tilde{R}_k &  = \frac{1}{K} \log_2\left(1+ \rho P_k x_k \right), \forall k \in \mathcal{S}_K\setminus \mathcal{S}.
\end{align}
The coefficients $|\mathcal{S}|/K$ and $1/K$ account for fair division of the resources between the users employing NOMA and OMA, respectively. Although in a standard NOMA network all users will employ NOMA, i.e., $\mathcal{S}$ and $\mathcal{S}_K$ coincide, in NOMA-R, $\mathcal{S}$ is a subset of $\mathcal{S}_K$. The formalization of the NOMA-R criterion stems from the requirement that $\mathcal{S}$ includes all users whose achievable rates are greater with NOMA than with OMA, i.e.,
\begin{align} \label{InitialNOMARelevantCriterionKUsers}
1+ \frac{\rho P_k x_{k}}{1+\sum_{\substack{j \in \mathcal{S}\\j<k}}\rho P_j x_{j}}  \geq  ( 1+\rho P_k x_k )^{\frac{1}{|\mathcal{S}|}}.
\end{align}

Users in $\mathcal{S}_K\setminus \mathcal{S}$ employ OMA.  Consequently, the data rate in NOMA-R is $\hat{R}_k = R_k  \ \forall k \in \mathcal{S}$ and $\hat{R}_k = \tilde{R}_k \ \forall k \notin \mathcal{S}$. In terms of implementation, the BS identifies the subsets $\mathcal{S}$ using its knowledge of the network's full channel state information (CSI). It then transmits to each user a one-bit feedback indicating whether they should use OMA or NOMA 
and either the user's index in the OMA subset or the NOMA cluster index if  several disjoint NOMA clusters are selected.  Therefore,  NOMA-R imposes at most one-bit signalling overhead with respect to OMA.

The EC in bits/s/Hz of user $k \in \mathcal{S}_K$ is defined as \cite{WuNegi2003}:
\begin{align}
    E_c^k  = -\frac{1}{\theta_k T_f B} \ln\left(\mathbb{E}[e^{-\theta_k T_f Br_k}] \right)   = \frac{1}{\beta_k} \log_2\left(\mathbb{E}[e^{\beta_k \ln(2) r_k}] \right) \nonumber
\end{align}
where $r_k$ is the achievable rate of user $k$ (equal to $R_k$, $\tilde{R}_k$, or $\hat{R}_k$ if the user employs NOMA, OMA or NOMA-R, respectively), $T_f$ is the symbol period and  $B$ is the occupied bandwidth. $\theta_k$, known as the QoS exponent\cite{WuNegi2003}, is the exponent of the exponential decay of the buffer overflow probability. Under a constant packet arrival rate assumption, the EC is defined as the maximum achievable rate such that a target delay-bound violation probability is met.  The more stringent the delay requirement, the larger the delay exponent $\theta_k$. 
To simplify the notation, we define $\beta_k = - \frac{\theta_k T_f B}{\ln(2)}$ as the negative QoS exponent. Closed form expressions of the EC when  OMA and NOMA are employed were derived in \cite{BelloChorti2019} for $K=2$ and are not repeated in the present for compactness.

\subsection{Probability of using NOMA while in  NOMA-R for $K\geq 2$}
When NOMA-R selects NOMA for $k \in \mathcal{S}$, the sum rate can easily be shown to be equal to $\frac{|\mathcal{S}|}{K}\log_2\left(1+ \sum_{k \in \mathcal{S}}\rho P_k x_{k} \right)$. Consequently, the subset 
$\mathcal{S}$ of $\mathcal{S}_K$ that maximizes the sum rate while satisfying (\ref{InitialNOMARelevantCriterionKUsers}) is selected by NOMA-R. Several disjoint clusters may also be selected if they independently verify (\ref{InitialNOMARelevantCriterionKUsers}). 

The probability of using NOMA when NOMA-R is employed, denoted by $\tau_K$ in the following, is the union of the probabilities to verify (\ref{InitialNOMARelevantCriterionKUsers}) for any subset $\mathcal{S} \subseteq \mathcal{S}_K$ with $|\mathcal{S}| \geq 2$ and its analytical derivation is very evolved. As an illustrative example, let us assume $\mathcal{S} = \{1,..., k\}$ with $k\leq K$. 
 Let  $y_k = \rho P_k x_{k}$ be  the weighted $k$th order statistics and $z_k =\sum_{j=1}^{k-1}\rho P_j x_{j}$ the weighted sum of the lowest $(k-1)$th order statistics. Then $\tau_k$ is equal to:
\begin{align} \label{ProbabilityKUsersNOMAR}
    \tau_k = Pr\left(\bigcap_{i=2:k} \left(\frac{(1+y_i) - (1+y_i)^{\frac{1}{k}}}{(1+y_i)^{\frac{1}{k}} - 1} \geq z_i \right) \right) .
\end{align}
 For the specific case $k=K$, the joint probability density function (pdf) of $(y_K, z_K)$ can be derived by using the moment generating function (MGF) of the weighted sum of the lowest order statistics, denoted as $\mathcal{M}_{z_K}$, as in \cite{yang_alouini_2011}. The pdf of $z_k$ can be obtained by using the following properties: $\mathcal{M}_{z_K} = \prod_{j=1}^{K-1} \mathcal{M}_{x_j}\left(\rho P_j x_{j} \right) $ and $\mathcal{L}^{-1} \left(  \mathcal{M}_{x_j}\left(\rho P_j x_{j} \right)\right)(t) = \frac{1}{\rho P_j} f_{x_j}\left(\frac{t}{\rho P_j} \right)$, where $\mathcal{L}^{-1} $ is the inverse Lagrange transform. Consequently, the joint pdf $f_{z_K,y_K = \bar{y}}$ can be derived from  that of $z_k$ and  $y_K$ in \cite[eq. (3.41)]{yang_alouini_2011}.
However, a closed-form expression of $Pr\left( \frac{(1+y_K) - (1+y_K)^{\frac{1}{K}}}{(1+y_K)^{\frac{1}{K}} - 1} \geq z_K  \right)$ cannot be obtained, and similarly to the conclusion in \cite[Section V.D]{KoAlouiniMRC07}, it should be calculated with a mathematical software.  Moreover when $k < K$, to the best of our knowledge, the joint pdf of $(y_K, z_K)$ is yet unknown. Consequently, 
when $K>2$, (\ref{ProbabilityKUsersNOMAR})  cannot be evaluated analytically with reasonable effort. For all these reasons, our analytical study is limited to  $K=2$ while we  provide numerical results  for $K>2$.

Finally, examining the case non i.i.d. channel coefficients, we note that the case of non-identical exponential distributions  can be treated as in \cite{Balakrishnan94}, while the case of non independent coefficients as in \cite{Inbar}. \textcolor{black}{If full CSI is not available at the BS, CSI uncertainties can be inserted in  users' distributions to derive an MA selection strategy \cite{Inbar}, and, when $K=2$,  (\ref{InitialNOMARelevantCriterionKUsers}) can be formulated as a binary hypothesis testing problem \cite{ShahD2DEC_19}.}

\subsection{Probability of using NOMA while in  NOMA-R for $K=2$}
In the following, we consider the two-users case and call user $2$ the strong user, and user $1$ the weak user. 
As $R_1  \geq \tilde{R}_1$ is always fulfilled, the NOMA-R strategy is used whenever  $R_2 \geq \tilde{R}_2$. The NOMA-R condition consequently simplifies to 
$x_2 \geq \frac{\rho^2 x_1^2 P_1^2-1}{\rho P_2}$ and $\tau_2=\tau(\rho)= Pr\left( x _2 \geq \frac{\rho^2 x_1^2 P_1^2-1}{\rho P_2} \right)$. Using the theory of order statistics, the pdf of $x_1$ is  $2 e^{-2 x_1}$, the pdf of $x_2$ is $2 e^{-x_2}(1- e^{-x_2})$ and the joint pdf of   $(x_1,x_2)$ is $2 e^{-x_1}e^{-x_2}$. Then $\tau(\rho)$ is equal to:
\begin{align} \label{TauExpression}
\tau &= \int_{x_1  = 0}^{\frac{P_2 + \sqrt{P_2^2+4 P_1^2}}{2 \rho P_1^2}} \int_{x_2 = x_1}^{+\infty} 2 e^{-x_1} e^{-x_2} dx_2 dx_1 \nonumber \\
&+ \int_{x_1 =\frac{P_2 + \sqrt{P_2^2+4 P_1^2}}{2 \rho P_1^2}}^{+\infty} \int_{x_2 = \frac{\rho^2 x_1^2 P_1^2-1}{\rho P_2}}^{+\infty} 2 e^{-x_1} e^{-x_2} dx_2 dx_1  \nonumber\\ 
&= f(\rho) + g(\rho)
\end{align}
where
\begin{align} \label{ffortau}
&f(\rho) = 1 - e^{-\frac{P_2 + \sqrt{P_2^2+4 P_1^2}}{ \rho P_1^2}}  \\ 
&g(\rho) =  \frac{\sqrt{\pi}e^{ \frac{4P_1^2+P_2^2}{4 \rho P_2 P_1^2}} \left(1- \mathrm{erf}\left(\frac{2 P_2 + \sqrt{P_2^2 + 4 P_1^2}}{2\sqrt{P_2 \rho} P_1}\right) \right) \sqrt{P_2} }{P_1 \sqrt{\rho}} 
\end{align}
and $\mathrm{erf}(x) = \frac{2}{\sqrt{\pi}} \int_0^x e^{-t^2}dt$. 
The boundary in both integrals in (\ref{TauExpression}) is due to the fact that $x_2$ should always be such that $x_2 \geq x_1$, but $\frac{\rho^2 x_1^2 P_1^2-1}{\rho P_2}$ is lower than $x_1$ if $x_1 \leq {\frac{P_2 + \sqrt{P_2^2+4 P_1^2}}{2 \rho P_1^2}}$. 
$\tau(\rho)$ is a monotonically decreasing function with respect to $\rho$.  (\ref{TauExpression}) is validated with Monte-Carlo simulations, shown in Fig. \ref{TauFigure}, assuming that $P_1+P_2=1$. 

\begin{lemma}
$\tau(\rho)$ tends to $1$ when $\rho$ tends to $0$ and  $\tau(\rho)$ tends to $0$ when   $\rho >> 1$. 
\end{lemma}

\begin{proof}  When $\rho \rightarrow 0$, $f(\rho) \rightarrow 1$. Moreover,  $\mathrm{erf}(x) \approx 1-e^{-x^2}/(x\sqrt{\pi})$ when $x>>1$, which implies that $g(\rho) \approx  a e^{-b/\rho}$ with $(a,b)$ two strictly positive constants. Therefore,  $g(\rho)\rightarrow 0$ when $\rho\rightarrow 0$. Furthermore, when $\rho>>1$, $f(\rho) \rightarrow 0$. As $\mathrm{erf}(x) \approx \frac{2}{\sqrt{\pi}} x e^{-x^2}$ when $x \rightarrow 0$, $g(\rho) \approx a_1 \frac{e^{a_2/\rho}}{\rho} + b_1 \frac{e^{b_2/\rho}}{\rho}$ where $(a_1,a_2,b_1,b_2)$ are constants %
and $g(\rho) \rightarrow 0$ 
\end{proof}

\begin{figure}[t]
\centering
  \includegraphics[width=2.9in]{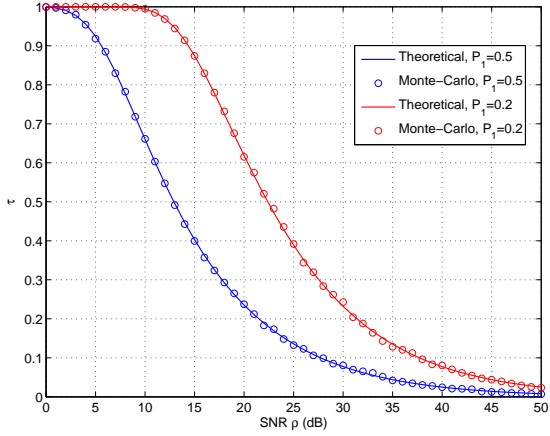}
  \vspace{-10pt}
  \caption{Validation of the closed-form expression of $\tau$ }\label{TauFigure}
\end{figure}

\section{NOMA-R Effective Rates}
\label{SectionAnalyticalExpressions}
The ECs of user $k=1,2$ when employing NOMA, OMA or NOMA-R are denoted by $E_{c,N}^k$, $E_{c,O}^k$ and $E_{c,R}^k$, respectively. 
\vspace{-0.5 cm}
\subsection{NOMA-R EC of user 1}
 We hereafter derive closed-form expressions of the EC with NOMA-R. 
As $\tau$ corresponds to the proportion of time spent in NOMA and $(1-\tau)$ is the proportion of time spent in OMA, the EC of user $k \in \{1,2\}$ when the NOMA-R strategy is employed is given as:
\begin{align} \label{NOMArelevantECGeneral}
    E_{c,R}^k & = \frac{1}{\beta_k} \log_2\left(\mathbb{E}\left[e^{\beta_k \tau R_k + \beta_k (1-\tau) \tilde{R}_k}\right]\right).
\end{align}

For user 1, the EC achieved with the NOMA-R strategy  is: 
\begin{align} \label{NOMArelevantECUser1}
    E_{c,R}^1 
   & = \frac{1}{\beta_1} \log_2\left(\mathbb{E}[ (1+\rho P_1 x_1)^{ \frac{\beta_1(\tau+1)}{2}} ]\right) \nonumber\\
   & =  \frac{1}{\beta_1} \log_2\left(\frac{2}{\rho P_1}  U\left(1,2+  \frac{\beta_1(\tau+1)}{2} , \frac{2}{\rho P_1}\right)  \right)
\end{align}
where $U(a,b,z) = \frac{1}{\Gamma(a)} \int_0^{\infty} e^{-zt}t^{a-1} (1+t)^{(b-a-1)} dt$ denotes the confluent hypergeometric function.

\begin{lemma}
 $E_{c,R}^1$ is monotonically increasing with  $\rho$.
 \end{lemma}
 
\begin{proof}
Let us consider $\rho_1$ and $\rho_2$ such that $\rho_1 < \rho_2$. For any value of $\beta_1$, let us define: $\beta_{1,a} =\frac{\beta_1(1+\tau(\rho_1))}{2}$ and  $\beta_{1,b} =\frac{\beta_1(1+\tau(\rho_2))}{2}$. Then from \cite[eq.(11)]{BelloChorti2019} and (\ref{NOMArelevantECUser1}), $E_{c,R}^1(\beta_1,\rho_1) =  E_{c,N}^1(\beta_{1,a}, \rho_1)$ and $ E_{c,R}^1(\beta_1,\rho_2) =  E_{c,N}^1(\beta_{1,b}, \rho_2)$. 
As $\tau(\rho)$ is a decreasing function with respect to $\rho$, and $\beta$ are negative,  $\beta_{1,a}  \leq \beta_{1,b}$. Moreover, $E_{c,R}^1(\beta,\rho)$ is increasing both with respect to $\beta$ and to $\rho$ according to \cite{BelloChorti2019}. Consequently, $E_{c,N}^1(\beta_{1,a}, \rho_1)  \leq  E_{c,N}^1(\beta_{1,a}, \rho_2) \leq E_{c,N}^1(\beta_{1,b}, \rho_2)$  and:
\begin{align} \label{CclMonotonicityEcr}
     E_{c,R}^1(\beta_1,\rho_1) \leq     E_{c,R}^1(\beta_1,\rho_2)   \hspace{0.5cm}  \forall \rho_1 < \rho_2
\end{align}
\end{proof}
\subsection{NOMA-R EC of user 2}
For user 2, the NOMA-R EC is given by
\begin{equation} \label{NOMArelevantECUse2}
    E_{c,R}^2 
 \!=\! \frac{1}{\beta_2}\! \log_2\left(\!\mathbb{E}\!\left[\left(1+ \frac{\rho P_2 x_2}{1+\rho P_1 x_1} \right)^{\beta_2 \tau} \!\!\!\! ( 1+\rho P_2 x_2 )^{\frac{\beta_2 (1-\tau)}{2}}\right]\!\right) 
\end{equation}

\begin{lemma}
When $\rho$ tends to $0$, the EC with the NOMA-R strategy becomes equivalent to that of NOMA given in \cite{BelloChorti2019}. 
\end{lemma}

\begin{proof}
When $\rho \rightarrow 0$,  $\tau(\rho) \rightarrow 1$ and therefore $( 1+\rho P_2 x_2 )^{\frac{\beta_2 (1-\tau)}{2}} \rightarrow 1$. Consequently,  $E_{c,R}^2$ tends to $E_{c,N}^2$. 
\end{proof}

\begin{lemma}
When $\rho>>1$, the EC with the NOMA-R strategy becomes equivalent to that of OMA and its closed-form expression is given by: 
\begin{align} \label{NOMArelevantECUse2NearInfinity}
    E_{c,R}^2 &  \approx \frac{1}{\beta_2} \log_2\left(\Gamma \left(\frac{\beta}{2} +1 \right) (\rho P_2)^{\frac{\beta}{2}} (2-2^{-\frac{\beta}{2} }) \right).
\end{align}
\end{lemma}

\begin{proof}
When $\rho>>1$,   $\left(1+ \frac{\rho P_2 x_2}{1+\rho P_1 x_1} \right)^{\beta_2 \tau} \rightarrow 1$ because $\tau \rightarrow 0$. Then using $(1+x)^{\alpha} \approx x^{\alpha}$,  the EC of user $2$ becomes:
\begin{align} 
    E_{c,R}^2 &   \approx \frac{1}{\beta_2} \log_2\left(\mathbb{E}\left[\left( \rho P_2 x_2 \right)^{\frac{\beta_2}{2}}\right]\right) \nonumber \\
    & \approx \frac{1}{\beta_2} \log_2\left(\int_{0}^{\infty}2 \left( \rho P_2 x_2 \right)^{\frac{\beta_2}{2}} e^{-x_2}(1-e^{-x_2}) dx_2 \right).\nonumber
\end{align}
The integral's closed-form expression leads to (\ref{NOMArelevantECUse2NearInfinity}). 
\end{proof}

\begin{theorem}
The EC of user $1$ is always larger with NOMA than with NOMA-R, while OMA is the worst strategy in terms of EC.
Moreover, the EC of user $2$ is always larger with the NOMA-R strategy than with NOMA or OMA.  
\end{theorem}

\begin{proof}
The NOMA-R instantaneous rate of user $1$  is equal to $\hat{R}_1 = \frac{(1+\tau)}{2} \log_2(1+\rho P_1 x_1)$, according to  (\ref{NOMArelevantECUser1}). Therefore
\begin{math} \label{ComparisonUser1InstantaneousCapacities1}
    \tilde{R}_1 \leq \hat{R}_1 \leq R_1
\end{math}
Then as $\beta_1$ is negative, 
  $ e^{\beta_1 R_1}  \leq e^{\beta_1\hat{R}_1} \leq   e^{\beta_1 \tilde{R}_1}$, and  $\mathbb{E}\left[e^{\beta_1 R_1}\right]  \leq  \mathbb{E}\left[e^{\beta_1\hat{R}_1}\right] \leq    \mathbb{E}\left[e^{\beta_1 \tilde{R}_1}\right]$, so that
\begin{align} \label{ECOrderingUser1}
    E_{c,N}^1 \geq  E_{c,R}^1 \geq E_{c,O}^1
\end{align}
The NOMA-R  rate of user $2$  is $\hat{R}_2= \max\{R_2,\tilde{R}_2\}$. Following the same steps as for user $1$, we conclude that
\begin{align} \label{OrderingEC2_third}
 E_{c,R}^2  \geq E_{c,N}^2
\text{  and }  E_{c,R}^2 \geq E_{c,O}^2
\end{align} 
\end{proof}

\textit{Remark}: $E_{c,R}^2$ asymptotically tends to  either $E_{c,N}^2$ or $E_{c,O}^2$, both of which are monotonically increasing with  $\rho$. Moreover, contrary to the NOMA strategy, the NOMA-R strategy does not lead to a saturation of the EC of the strong user because of (\ref{OrderingEC2_third}) and because $E_{c,O}^2$ increases without bound with $\rho$ \cite{BelloChorti2019}.

\section{Numerical Results \label{SectionPerformance}}
Performance comparisons with respect to EC with the NOMA-R strategy are provided for $P=[0.2, 0.8]$, $P=[0.05, 0.15,  0.8]$ and $P=[0.01, 0.04,0.15,0.8]$ when $K=2, 3$ and $4$, respectively. Unless otherwise stated, $\beta_k =  -2, \ \forall k$.   Fig. \ref{ECapacityPerUser22} validates the results of Theorem 1. 
Fig. \ref{SumEC2and3Users} 
shows that the largest sum EC is always achieved  with NOMA-R whatever the value of $K$. Moreover, the sum EC with NOMA-R coincides with that of NOMA when the SNR is lower than a minimum value $\rho_{\text{min}}$ that increases with  $K$, because constraint $(3)$ becomes less stringent when $K$ increases.
\begin{figure}[t]
\centering
  \includegraphics[width=2.9in]{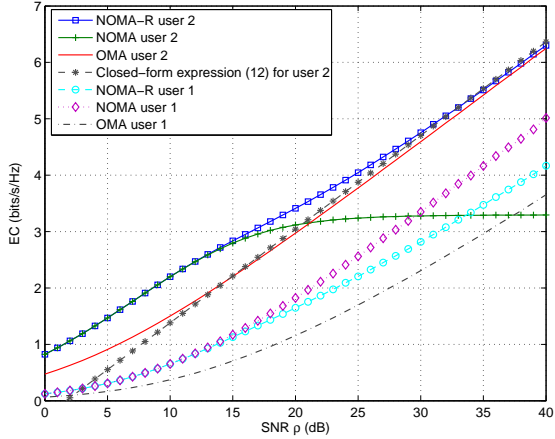}
     \vspace{-10pt}
  \caption{EC per user versus SNR $\rho$ when $K=2$}\label{ECapacityPerUser22}
\end{figure}
\begin{figure}[t]
\centering
  \includegraphics[width=2.9in]{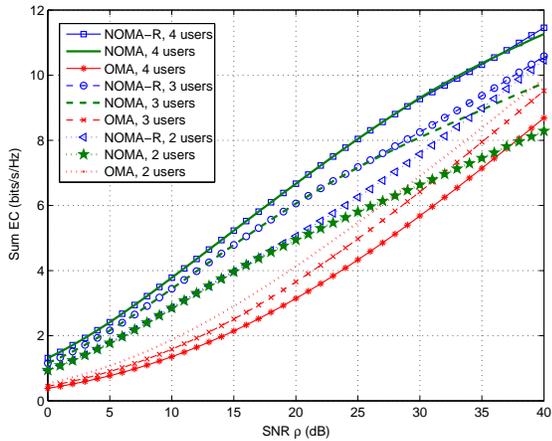}
     \vspace{-10pt}
  \caption{Sum EC versus SNR $\rho$}\label{SumEC2and3Users}
\end{figure}
Fig. \ref{SumEC23UsersDependingBeta}  shows the dependency of the EC on to $\beta_1$ when the SNR is equal to $35$ dB, $\beta_K = -2$ and  $\beta_j = \beta_1, \ \forall j<K$. The sum EC is larger with NOMA-R, except when $\beta_1$ approaches $0$.  
The NOMA-R strategy is consequently  more favorable when the target delay-bound violation probabilities are  more stringent, especially for weak users. 
\begin{figure}[t]
\centering
 \includegraphics[width=2.9in]{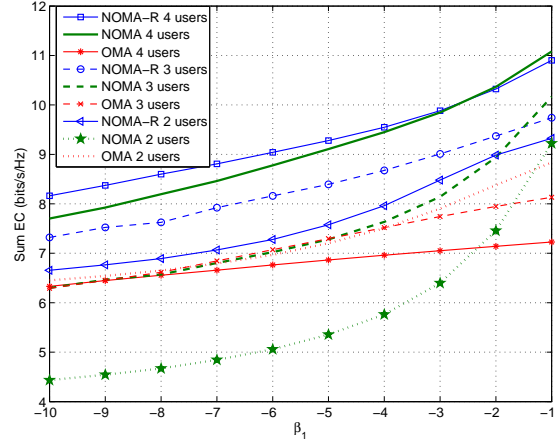}
     \vspace{-10pt}
  \caption{Sum EC versus $\beta_1$}\label{SumEC23UsersDependingBeta}
\end{figure}

\section{Conclusions}
In this letter the EC performance of an adaptive MA strategy,  NOMA-R, was studied both analytically for $K=2$ users and numerically for larger values of $K$. It was shown that NOMA-R is an advantageous strategy for delay constrained applications in B5G, e.g., URLLC, particularly as the users' delay-outage probability constraints become more stringent.
\bibliographystyle{IEEEbib}

\bibliography{NOMAEC_biblio.bib}
\end{document}